\documentclass[conference]{IEEEtran}
\usepackage{cite}
\usepackage{amsmath,amssymb,amsfonts,amsthm}
\usepackage{graphicx}
\usepackage{booktabs}
\usepackage{multirow}

\newtheorem{theorem}{Theorem}
\newtheorem{lemma}{Lemma}

\newtheorem{corollary}{Corollary}
\newtheorem{assumption}{Assumption}
\newtheorem{definition}{Definition}

\makeatletter
\def\thm@space@setup{\thm@preskip=2pt plus 1pt minus 1pt\thm@postskip=2pt plus 1pt minus 1pt}
\makeatother

\begin{document}

\title{Psychoacoustically Aligned Latent Smoothing for Adversarial Robustness of Full-Duplex Speech-to-Speech Dialogue Models}

\author{\IEEEauthorblockN{Kian Shamsaie, Iman Modarressi}
\IEEEauthorblockA{People Make Things\\
\textit{\{k,i\}@peoplemakethings.com}}}
\maketitle

\begin{abstract}
End-to-end speech-to-speech dialogue models listen and speak simultaneously, so a continuously open acoustic channel is exposed to adversarial manipulation. We formalize imperceptible attacks on full-duplex agents as optimization over additive perturbations confined beneath the psychoacoustic masking threshold of the carrier speech, under three goals: targeted semantic hijacking, response suppression, and policy jailbreaking. Against an undefended Moshi-style agent, white-box attacks succeed in up to 91.7\% of trials. We then introduce psychoacoustically aligned latent smoothing (PALS), which injects anisotropic Gaussian noise shaped by local codebook covariance at the residual-vector-quantized latent interface, with input noise shaped by the masking threshold constraining the attacker and trained by a Kullback--Leibler consistency objective. Deployed with no inference-time cost, PALS reduces hijack to 8.3\%, mute to 11.2\%, and jailbreak to 9.1\% at clean quality within 2.3\%. A Monte Carlo-smoothed variant certifies an ellipsoidal latent radius up to 0.616, a guaranteed floor that the empirical robustness far exceeds.
\end{abstract}


\section{Introduction}
\label{sec:intro}

Full-duplex speech-to-speech (S2S) dialogue models such as Moshi \cite{defossez2024moshi}, dGSLM \cite{nguyen2023generative}, and synchronous large language model (LLM) agents \cite{veluri2024beyond} encode the user channel into discrete tokens with a streaming residual vector quantization (RVQ) codec \cite{zeghidour2022soundstream,defossez2023high}, model both conversational streams jointly, and speak with sub-second latency. The shift creates an attack surface more dangerous than that of automatic speech recognition (ASR): a full-duplex agent ingests every frame of an always-open microphone, maintains persistent dialogue state an adversary can poison early and exploit late, and converts audio directly into spoken actions without an inspectable transcript. ASR is reliably fooled by inaudible perturbations, via waveform optimization \cite{carlini2018audio}, masking-threshold hiding \cite{schonherr2019adversarial,qin2019imperceptible}, or universal perturbations \cite{neekhara2019universal}, and speech-capable LLMs inherit these vulnerabilities \cite{peri2024speechguard,kang2025advwave,shen2024voice}. Whether, and how cheaply, a full-duplex S2S agent can be protected remains open.

This paper gives a twofold answer. First, we port the psychoacoustically masked threat model to the full-duplex setting and show that it is devastating there. The attacker maximizes an objective over additive waveform perturbations confined to a perceptual ball whose per-bin radii derive from the masking threshold of the carrier conversation, under three goals: a targeted semantic hijack forcing an attacker-chosen response, a denial attack muting the agent, and a policy jailbreak delivered through an adversarial audio prefix. Digitally and white-box, these attacks succeed against an undefended Moshi-style agent in 91.7\%, 88.4\%, and 76.2\% of trials, consistent with rates against ASR \cite{carlini2018audio,qin2019imperceptible} and speech LLMs \cite{peri2024speechguard}, and remain effective over the air \cite{yakura2019robust}.

Second, and centrally, we propose a defense whose distinguishing property is minimality. \emph{Psychoacoustically aligned latent smoothing} (PALS), applied during ordinary post-training at the RVQ latent interface, perturbs each pre-quantization latent frame with anisotropic Gaussian noise whose covariance is the local codebook covariance, so smoothing respects the geometry quantization imposes, and re-quantizes through a straight-through estimator so the model always consumes valid codes. This is coupled to a matched input-space augmentation: noise whose spectral envelope is the masking threshold that constrains the attacker, which we show is the maximum-entropy distribution under the threshold-implied covariance constraint and hence the least presumptive model of an attacker confined to the perceptual ball. Training adds a single Kullback--Leibler (KL) consistency term between output distributions under clean and augmented latents, in the spirit of TRADES \cite{zhang2019theoretically}. One scalar $\sigma$ scales both augmentations; no attack is run during training and no architecture changes.

We separate two operating points and never conflate them. PALS-Train is the deployed system: the post-trained agent runs unmodified, with no inference-time smoothing or cost, and is evaluated empirically; it is the system behind every attack-success number, including the headline hijack rate of 8.3\%. PALS-Certify is a Monte Carlo-smoothed variant whose decision is aggregated over $n=1{,}000$ noise draws at inference; it carries a provable certificate at the smoothing cost, and supplies every certified-accuracy number. The defense comes with theory under explicit assumptions and inline proofs. Adapting randomized smoothing \cite{cohen2019certified} to quantized latents with a fixed certification-time covariance, we show that a sufficient top-two posterior margin keeps the decision of PALS-Certify constant within an ellipsoidal (Mahalanobis) latent radius; because smoothing needs no continuity of the base function, the hard quantizer is certified for free. A TRADES-style decomposition of the smoothed loss bounds adversarial risk by clean risk, the trained consistency term, a coverage gap, and a tail penalty, and a corollary transfers the latent certificate to the masking-weighted norm. Empirically, PALS-Train cuts targeted hijack from 91.7\% to 8.3\% at matched perceptibility for a 2.3\% relative quality drop, dominating Gaussian augmentation, input randomized smoothing \cite{olivier2021sequential}, projected gradient descent (PGD) adversarial training \cite{madry2018towards}, and a Demucs front-end \cite{defossez2020real}. The certificate is honestly a small-radius floor: PALS-Certify certifies 47.0\% of decisions at latent radius $0.5$, capped at $0.616$ by the protocol and corresponding to $\Delta\approx-26$ dB, far below the $0$ dB attacks, so it is a guaranteed floor, not a guarantee against the strong attacks.

\section{Related Work}
\label{sec:related}

Adversarial examples were first characterized for visual classifiers \cite{szegedy2014intriguing}; Carlini and Wagner showed targeted attacks transfer to speech-to-text with near-imperceptible perturbations \cite{carlini2018audio}. Psychoacoustic hiding constrained the perturbation beneath an MPEG-style masking threshold \cite{schonherr2019adversarial}, and Qin et al.\ added expectation over room transformations for imperceptible, physically robust targeted attacks \cite{qin2019imperceptible}. Universal perturbations \cite{neekhara2019universal} and physical playback \cite{yakura2019robust} establish that the threat persists outside the digital sandbox; a systematization appears in \cite{abdullah2021sok}, and universal segments that mute speech foundation models \cite{raina2024muting} foreshadow the denial attack we formalize for full-duplex agents.

Text-domain adversarial suffixes \cite{zou2023universal} motivated audio-domain attacks on speech LLMs: SpeechGuard reported white-box jailbreak success near 90\% \cite{peri2024speechguard}, AdvWave introduced stealthy adaptive-target jailbreak optimization \cite{kang2025advwave}, voice-mode jailbreaks against commercial assistants exploit narrative framing \cite{shen2024voice}, and benchmark suites standardize such evaluations \cite{peng2025jalmbench}. These works target half-duplex audio LLMs \cite{tang2024salmonn,chu2024qwen2audio} or proprietary assistants; none formalizes the masked-perturbation attacker for streaming full-duplex S2S models, and none offers a certified defense at the codec latent interface, the gap addressed here.

On the defense side, adversarial training with PGD \cite{madry2018towards} and its TRADES refinement \cite{zhang2019theoretically} remain the strongest empirical defenses but multiply training cost and demand adaptive evaluation \cite{athalye2018obfuscated,tramer2020adaptive}, while audio purification, from perceptual input transformations \cite{hussain2021waveguard} to diffusion purification \cite{wu2023defending}, can be bypassed by attackers that differentiate through the purifier \cite{athalye2018obfuscated}. Certified approaches derive from randomized smoothing \cite{cohen2019certified}, with antecedents in differential privacy \cite{lecuyer2019certified}, strengthened by adversarially trained smoothed classifiers \cite{salman2019provably} and denoised smoothing \cite{salman2020denoised}; sequential randomized smoothing certified ASR transcripts under inaudible-noise budgets at a measurable cost in word error rate (WER) \cite{olivier2021sequential}. The use of input-dependent and anisotropic noise is itself an established line: data-dependent smoothing optimizes the noise variance per input \cite{alfarra2022data}, ANCER fits per-sample axis-aligned ellipsoidal certificates by volume maximization \cite{eiras2021ancer}, and center smoothing certifies structured outputs \cite{kumar2021center}; latent-space smoothing with an orthogonal, known-Lipschitz encoder propagates a latent certificate back to the input \cite{zeng2021latent}. Our certified variant is a direct instance of this anisotropic, latent-space family: Theorem~\ref{thm:certificate} is the standard whitening-to-isotropic-Cohen reduction composed with a Lipschitz pullback. Its genuine novelty is not the certificate machinery but the smoothing law, a psychoacoustically matched covariance tied to the imperceptibility threat model, derived by codebook geometry in the latent space and by a maximum-entropy argument \cite{jaynes1957information} under the masking-threshold covariance in the input space, so the defender's randomization matches exactly the attacker the masking model permits.

Generative spoken dialogue modeling on dual channels began with dGSLM \cite{nguyen2023generative}; Moshi unified a text LLM backbone with the streaming Mimi codec into a real-time full-duplex agent \cite{defossez2024moshi}, building on neural codec language modeling \cite{borsos2023audiolm,zeghidour2022soundstream,defossez2023high}. Synchronous LLMs \cite{veluri2024beyond}, low-latency cascades \cite{fang2025llamaomni}, and end-to-end chatbots \cite{zeng2024glm4voice} broaden the design space, while benchmarks quantify turn-taking \cite{lin2025fullduplexbench} and corruption sensitivity \cite{shah2025speech}. The dyadic corpora enabling such training \cite{reece2023candor,sheikh2025scalable,agrawal2025seamless,carletta2005ami,otospeech2025full} also supply the carrier conversations on which imperceptible attacks are mounted.

\begin{figure}[t]
\centering
\includegraphics[width=0.80\columnwidth]{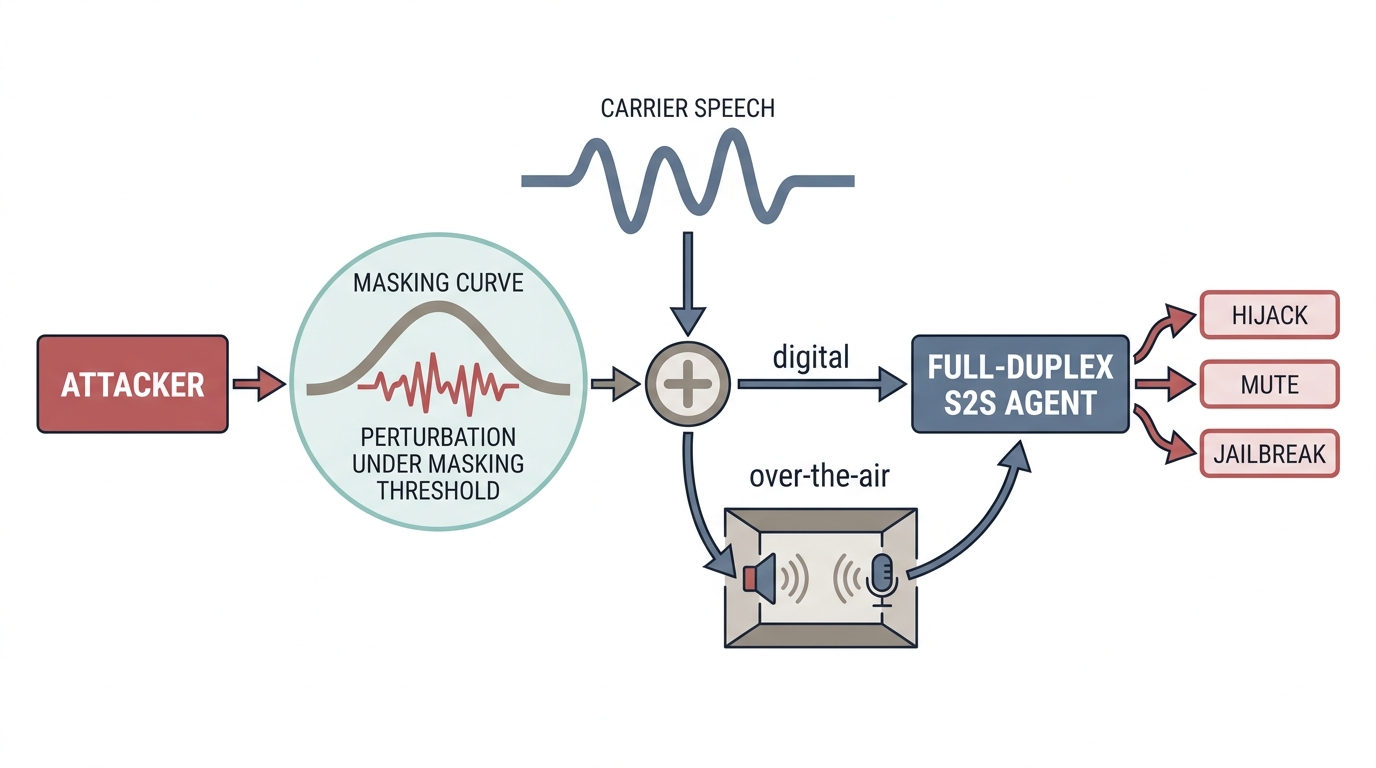}
\caption{Threat model. The attacker optimizes a perturbation confined beneath the masking threshold of the carrier speech, delivers it digitally or over the air, and pursues hijack, mute, or jailbreak goals.}
\label{fig:threat}
\end{figure}

\section{Threat Model and Attack Formulation}
\label{sec:threat}

\subsection{System Under Attack}
Let $x \in \mathbb{R}^{T}$ denote the user-channel waveform at 24 kHz. A streaming convolutional encoder $E_c : \mathbb{R}^{T} \to \mathbb{R}^{F' \times d}$ produces continuous latent frames $z = E_c(x)$ at 12.5 Hz with $d = 512$, which an RVQ quantizer $Q$ with eight codebooks of size 2048 maps to codes $\hat z = Q(z)$, following the Mimi design \cite{defossez2024moshi,zeghidour2022soundstream}. A dialogue model $f_\vartheta$ consumes the interleaved user and agent token streams and an inner-monologue text stream; $p_\vartheta(\cdot \mid Q(E_c(x)), h)$ is the distribution over agent responses given history $h$. A response is summarized by a semantic decision in a finite set $\mathcal{Y}$ induced by an external judge (Section~\ref{sec:setup}), with $h_{\mathrm{sem}}$ mapping a latent sequence to the class of the response it elicits. Figure~\ref{fig:threat} depicts the attack surface.

\subsection{The Perceptual Ball}
Following \cite{qin2019imperceptible,schonherr2019adversarial}, let $D$ denote the short-time Fourier transform (STFT) with frames $f = 1, \dots, F$ and bins $k = 1, \dots, K$, and let $\theta_x(f,k) > 0$ be the global masking threshold of carrier $x$ computed by MPEG-1 psychoacoustic model~1 \cite{painter2000perceptual,fastl2007psychoacoustics}. The attacker chooses $\delta$ in the \emph{perceptual ball}
\begin{equation}
\label{eq:ball}
\mathcal{P}_x(\eta) = \bigl\{ \delta \in \mathbb{R}^{T} : \lvert (D\delta)_{f,k} \rvert^{2} \le \eta^{2}\, \theta_x(f,k) \;\; \forall (f,k) \bigr\},
\end{equation}
equivalently the ball of radius $\eta$ in the masking-weighted supremum norm $\lVert \delta \rVert_{\theta_x,\infty} = \max_{f,k} \lvert (D\delta)_{f,k} \rvert / \sqrt{\theta_x(f,k)}$. We report the budget as $\Delta = 20 \log_{10} \eta$ dB; $\Delta = 0$ touches the threshold of audibility and $\Delta < 0$ buys margin. Projection onto the convex ball is per-bin magnitude clipping followed by inverse STFT.

\subsection{Attack Objectives}
The attacker has white-box access to $E_c$, $Q$, and $f_\vartheta$, full knowledge of the defense, and optimizes with Adam under projection onto \eqref{eq:ball}, gradients passing through $Q$ by the straight-through estimator as required for adaptive evaluation \cite{athalye2018obfuscated,tramer2020adaptive}. The targeted semantic hijack forces an attacker-chosen response $y^{\star}$,
\begin{equation}
\label{eq:hijack}
\min_{\delta \in \mathcal{P}_x(\eta)} \; -\log p_\vartheta\bigl(y^{\star} \mid Q(E_c(x + \delta)), h\bigr),
\end{equation}
generalizing targeted ASR attacks \cite{carlini2018audio,qin2019imperceptible} to spoken responses. The mute attack drives the agent stream to its silence tokens by minimizing $-\sum_{t} \log p_\vartheta(s_{\varnothing}^{(t)} \mid Q(E_c(x + \delta)), h)$ over $\delta \in \mathcal{P}_x(\eta)$, with $s_{\varnothing}^{(t)}$ the natural-pause code, the full-duplex analogue of muting speech foundation models \cite{raina2024muting}. The jailbreak attack optimizes a reusable prefix $a$ prepended to harmful spoken requests $x_j$ from an AdvBench-derived set \cite{zou2023universal,kang2025advwave} toward affirmative openings, $\min_{a \in \mathcal{P}_{a_0}(\eta)} \sum_j -\log p_\vartheta(y^{+}_j \mid Q(E_c(a \oplus x_j)), h)$. Over-the-air attacks wrap the objectives in an expectation over room impulse responses and noise \cite{qin2019imperceptible,yakura2019robust}.

\section{Psychoacoustically Aligned Latent Smoothing}
\label{sec:method}

PALS perturbs the model where the attack must pass, the latent bottleneck of the codec, with a single scalar $\sigma$ scaling the entire defense; Figure~\ref{fig:pipeline} summarizes the pipeline.

\begin{figure}[t]
\centering
\includegraphics[width=0.80\columnwidth]{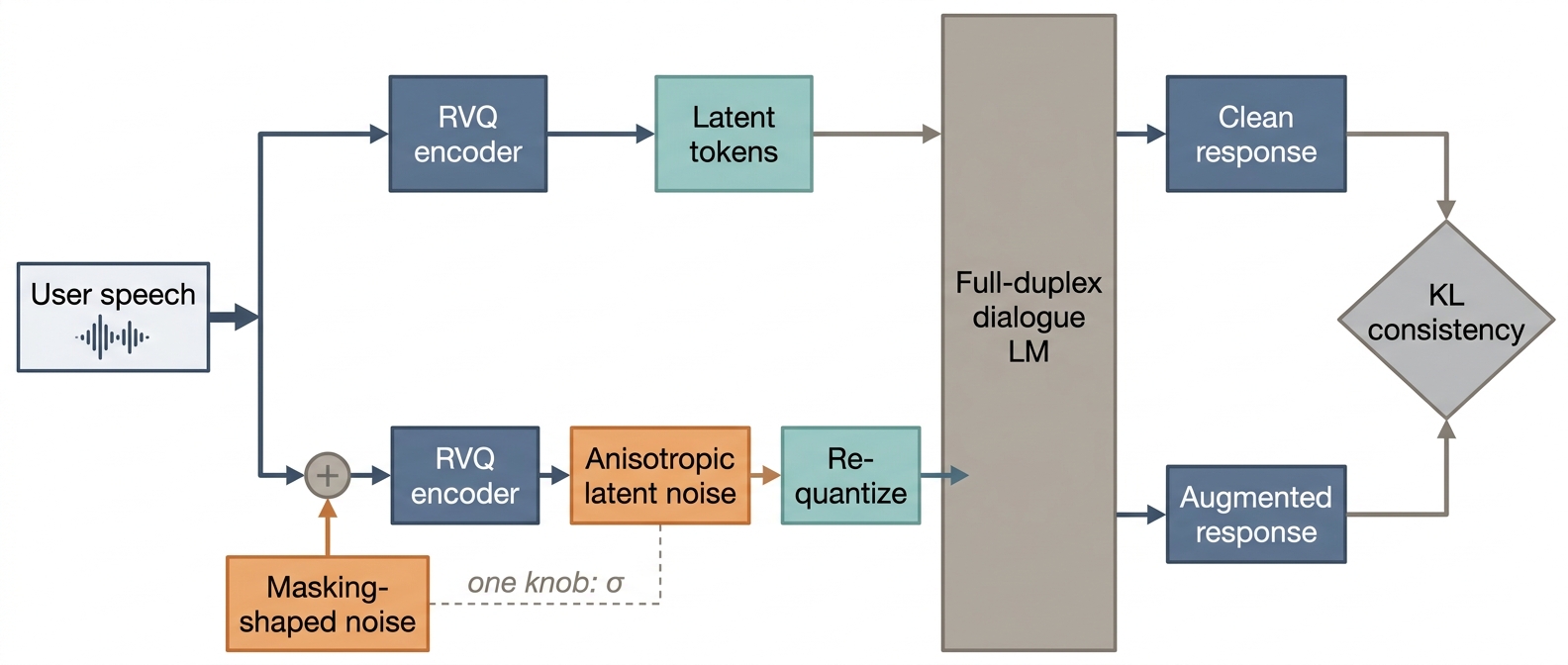}
\caption{The PALS post-training pipeline. The augmented path adds masking-shaped noise in the waveform domain and anisotropic codebook-covariance noise in the latent domain, re-quantizes through a straight-through estimator, and is tied to the clean path by a KL consistency term. One scalar $\sigma$ governs both augmentations.}
\label{fig:pipeline}
\end{figure}

\subsection{Geometry-Aware Latent Augmentation}
RVQ imposes a strongly anisotropic latent geometry: high intra-cell variance lies along directions where encoder outputs fluctuate, exactly where an attack that survives quantization must cross cell boundaries. During calibration we assign each latent frame to its nearest first-codebook entry $c_v$ and accumulate the empirical covariance $\widehat\Sigma_v$ of the residuals in cell $v$. The smoothing covariance for a frame $z_t$ in cell $v(z_t)$ is the shrunk, trace-normalized matrix
\begin{equation}
\label{eq:cov}
\Sigma(z_t) = \tfrac{d}{\operatorname{tr}(\Sigma'_t)} \Sigma'_t, \qquad \Sigma'_t = (1 - \rho)\, \widehat\Sigma_{v(z_t)} + \rho I_d,
\end{equation}
with shrinkage $\rho = 0.1$ guaranteeing $\Sigma(z_t) \succ 0$. The augmented latent is $\tilde z_t = z_t + \sigma\, \varepsilon_t$, $\varepsilon_t \sim \mathcal{N}(0, \Sigma(z_t))$ per frame, and the model consumes re-quantized codes $Q(\tilde z)$ with straight-through gradients, keeping training on the discrete support seen at inference.

\subsection{The Masking-Matched Input Dual}
To also train the encoder against threshold-shaped energy, we add a waveform-domain augmentation $\nu$ whose STFT coefficients are independent zero-mean circularly symmetric complex Gaussians with threshold-matched variance, $(D\nu)_{f,k} \sim \mathcal{CN}\bigl(0, \eta(\sigma)^2\, \theta_x(f,k)\bigr)$, synthesized by inverse STFT. The gain $\eta(\sigma) = \kappa \sigma$ is calibrated once so the median latent displacement matches that injected directly in latent space, leaving $\sigma$ the only free parameter. The masking threshold defines a per-bin expected-power (covariance) budget $\mathbb{E}\lvert(D\nu)_{f,k}\rvert^2 \le \eta^2\theta_x(f,k)$, and Lemma~\ref{lem:maxent} identifies the noise maximally noncommittal under it. This is a covariance constraint, not the pointwise supremum constraint \eqref{eq:ball} defining the perceptual ball; the Gaussian places mass outside the ball with probability near one, so we relate the two through Corollary~\ref{cor:transfer} rather than as an identity.

\begin{lemma}[Maximum entropy under the masking-threshold covariance]
\label{lem:maxent}
Among all distributions on $\nu$ whose STFT coefficients are independent across bins, zero-mean, and satisfy the expected-power (covariance) constraint $\mathbb{E}\,\lvert (D\nu)_{f,k} \rvert^{2} \le \eta^{2} \theta_x(f,k)$ for all $(f,k)$, differential entropy is maximized uniquely by the circularly symmetric complex Gaussian with per-bin variance exactly $\eta^{2}\theta_x(f,k)$, i.e., by the masking-shaped noise used by PALS.
\end{lemma}
\begin{proof}
For complex $u$ with $\mathbb{E}\lvert u \rvert^{2} = s$, $h(u) \le \log(\pi e s)$ with equality iff $u \sim \mathcal{CN}(0, s)$. By independence the joint entropy is the sum of per-bin entropies, each bounded by $\log(\pi e \eta^{2} \theta_x(f,k))$ and increasing in the variance, so the bound is attained by saturating each covariance constraint with a Gaussian; uniqueness follows from strict concavity. The optimization is over the second moment only, hence a statement about the covariance, not pointwise feasibility in \eqref{eq:ball}.
\end{proof}

By Jaynes's principle \cite{jaynes1957information}, a defender who knows only the attacker's per-bin power budget can do no better than randomize with this least informative distribution, which covers the perceptual ball in the weighted-$\ell_2$ sense made precise below.

\begin{table*}[t]
\caption{Main results at attack budget $\Delta = 0$ dB. The attack, clean, and over-the-air columns evaluate the empirically deployed operators (PALS row $=$ PALS-Train, no inference-time smoothing); the final column reports certified semantic accuracy of the Monte Carlo-smoothed PALS-Certify at latent $\ell_2$ radius 0.5 (inside the protocol ceiling 0.616). Attack rates carry bootstrap 95\% confidence half-widths.}
\label{tab:main}
\centering
\scriptsize
\setlength{\tabcolsep}{4pt}
\renewcommand{\arraystretch}{0.86}
\begin{tabular}{@{}lcccccccc@{}}
\toprule
Defense & Hijack $\downarrow$ & Mute $\downarrow$ & Jailbreak $\downarrow$ & OTA hijack $\downarrow$ & UTMOS $\uparrow$ & WER $\downarrow$ & TT-F1 $\uparrow$ & Cert.@0.5 $\uparrow$ \\
\midrule
Undefended & 91.7\,$\pm$\,1.7 & 88.4\,$\pm$\,2.0 & 76.2\,$\pm$\,3.7 & 46.3\,$\pm$\,3.1 & 3.91 & 7.4 & 0.71 & --- \\
Gaussian augmentation & 41.5\,$\pm$\,3.1 & 39.8\,$\pm$\,3.0 & 33.0\,$\pm$\,4.1 & 21.7\,$\pm$\,2.6 & 3.78 & 8.3 & 0.69 & --- \\
Demucs front-end \cite{defossez2020real} & 53.7\,$\pm$\,3.1 & 49.2\,$\pm$\,3.1 & 44.6\,$\pm$\,4.3 & 18.9\,$\pm$\,2.4 & 3.84 & 7.9 & 0.70 & --- \\
Input smoothing \cite{olivier2021sequential} & 28.9\,$\pm$\,2.8 & 26.1\,$\pm$\,2.7 & 24.4\,$\pm$\,3.7 & 15.2\,$\pm$\,2.2 & 3.62 & 9.8 & 0.65 & 31.5 \\
PGD-AT \cite{madry2018towards} & 19.4\,$\pm$\,2.5 & 17.2\,$\pm$\,2.3 & 15.8\,$\pm$\,3.2 & 9.6\,$\pm$\,1.8 & 3.70 & 8.9 & 0.69 & --- \\
PALS (ours, $\sigma = 0.25$) & 8.3\,$\pm$\,1.7 & 11.2\,$\pm$\,1.9 & 9.1\,$\pm$\,2.5 & 4.7\,$\pm$\,1.3 & 3.82 & 7.7 & 0.70 & 47.0$^{\dagger}$ \\
\multicolumn{9}{@{}l@{}}{\scriptsize $^{\dagger}$ PALS-Certify (inference-time MC smoothing, $n=1{,}000$); all other columns are PALS-Train.}\\
\bottomrule
\end{tabular}
\end{table*}

\subsection{Training Objective}
Let $\mathcal{L}_{\mathrm{post}}(\vartheta; x)$ be the standard post-training next-token cross-entropy over the agent audio and inner-monologue text streams \cite{defossez2024moshi}. Writing $\hat z = Q(E_c(x))$ for the clean codes and $\tilde z = Q\bigl(E_c(x + \nu) + \sigma \varepsilon\bigr)$ for the doubly augmented codes, PALS minimizes
\begin{equation}
\label{eq:loss}
\mathcal{L}(\vartheta) = \mathbb{E}\Bigl[ \mathcal{L}_{\mathrm{post}}(\vartheta; x) + \beta\, \mathrm{KL}\bigl( p_\vartheta(\cdot \mid \hat z, h) \,\Vert\, p_\vartheta(\cdot \mid \tilde z, h) \bigr) \Bigr],
\end{equation}
with $\beta = 4$, the TRADES coupling \cite{zhang2019theoretically} computed per emitted token and averaged. No adversarial optimization runs during training, so the overhead is one extra forward pass, roughly $1.4\times$ wall-clock against $3.4\times$ for ten-step PGD adversarial training. The trained model defines the deployed operator PALS-Train, which runs unmodified at inference with no added cost and is evaluated empirically. A second operator, PALS-Certify, aggregates the decision over $n$ Monte Carlo draws of $\varepsilon$ from a fixed certification-time covariance; it pays an $n$-fold inference cost and is the object of the certificate analyzed next. The two operators share weights but differ at inference, and we keep their results separate throughout.

\begin{figure*}[t]
\centering
\includegraphics[width=0.80\textwidth]{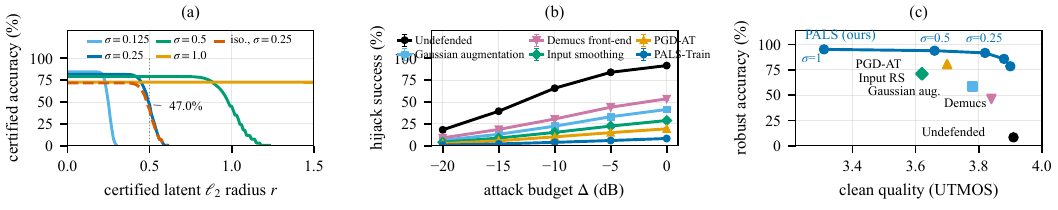}
\caption{(a) PALS-Certify: certified semantic accuracy versus latent $\ell_2$ radius (isotropic ablation dashed; $n = 1{,}000$ draws, $\alpha = 10^{-3}$). Each $\sigma$ curve terminates at its own Clopper--Pearson/trace ceiling $R_\Sigma^{\max}(\sigma)$, $0.616$ at $\sigma=0.25$; the data come from \texttt{certify.py}, which enforces the cap. (b) PALS-Train: digital white-box hijack success versus budget $\Delta$ with bootstrap 95\% confidence intervals. (c) PALS-Train: empirical robustness versus clean quality; the $\sigma$ sweep dominates all baselines, with the knee at $\sigma = 0.25$.}
\label{fig:results}
\end{figure*}

\section{Certified Robustness Analysis}
\label{sec:theory}

Throughout, $\Phi$ is the standard normal cumulative distribution function and $\Phi^{-1}$ its inverse. The certificate is a property of PALS-Certify, whose smoothing law is identical at the clean point and every perturbation of it.

\begin{definition}[PALS-Certify, fixed-covariance smoothing]
\label{def:certify}
To certify the decision at a clean latent point $z$, fix the covariance $\Sigma := \Sigma(z) \succ 0$ at the cell of $z$ (Eq.~\eqref{eq:cov}), or a single globally whitened covariance, and hold it constant over $\mathbb{R}^m$, drawing $\varepsilon \sim \mathcal{N}(0, \sigma^2 \Sigma)$; the certified operator is $\tilde g(w) = \arg\max_y \mathbb{P}(h(w + \varepsilon) = y)$.
\end{definition}

PALS-Train instead reselects a cell-dependent $\Sigma(z_t)$ per frame, so its noise law at $z$ differs from that at $z+\delta_z$ and the smoothing premise fails; freezing $\Sigma$ is therefore necessary for a valid certificate, and is what the released harness runs.

\begin{theorem}[Anisotropic smoothing certificate on quantized latents]
\label{thm:certificate}
Let $h : \mathbb{R}^{m} \to \mathcal{Y}$ be any measurable function, let $\Sigma \succ 0$ be \emph{fixed}, and let $\varepsilon \sim \mathcal{N}(0, \sigma^{2}\Sigma)$ as in Definition~\ref{def:certify}. Suppose for some class $A$ and bounds $\underline{p_A}, \overline{p_B} \in [0,1]$,
\begin{equation}
\mathbb{P}(h(z + \varepsilon) = A) \ge \underline{p_A} \ge \overline{p_B} \ge \max_{y \ne A} \mathbb{P}(h(z + \varepsilon) = y).
\end{equation}
Then $\tilde g(z + \delta_z) = A$ for every latent perturbation $\delta_z$ inside the ellipsoid
\begin{equation}
\label{eq:radius}
\lVert \Sigma^{-1/2} \delta_z \rVert_2 \le \tfrac{\sigma}{2} \bigl( \Phi^{-1}(\underline{p_A}) - \Phi^{-1}(\overline{p_B}) \bigr) =: R_{\Sigma}.
\end{equation}
A Euclidean ball gives the weaker guarantee $\tilde g(z + \delta_z) = A$ whenever $\lVert \delta_z \rVert_2 \le R := \sqrt{\lambda_{\min}(\Sigma)}\, R_{\Sigma}$.
\end{theorem}
\begin{proof}
Because $\Sigma$ is fixed, the whitening map $W = \Sigma^{-1/2}$ does not depend on the evaluation point. With $u = Wz$ and $h_W(u) = h(W^{-1}u)$, the noise $W\varepsilon \sim \mathcal{N}(0, \sigma^{2} I)$ is isotropic and $h_W$ measurable. The Neyman--Pearson argument of Cohen et al.\ \cite{cohen2019certified} certifies $h_W$ within isotropic radius $\frac{\sigma}{2}(\Phi^{-1}(\underline{p_A}) - \Phi^{-1}(\overline{p_B}))$ around $u$; a latent perturbation $\delta_z$ maps to $W\delta_z$ with $\lVert W\delta_z\rVert_2 = \lVert \Sigma^{-1/2}\delta_z\rVert_2$, which is \eqref{eq:radius}, and the $\ell_2$ form follows from $\lVert \Sigma^{-1/2}\delta_z \rVert_2 \le \lVert \delta_z \rVert_2 / \sqrt{\lambda_{\min}(\Sigma)}$. The argument needs no continuity of $h$, so it applies verbatim to $h = h_{\mathrm{sem}} \circ Q$, absorbing the hard quantizer and decoder into the base function.
\end{proof}

The certificate is natively the ellipsoid \eqref{eq:radius}; the $\ell_2$ ball is its inscribed sphere and the only honest scalar summary under anisotropic noise. Trace normalization fixes $\operatorname{tr}\Sigma = d$, so $\lambda_{\min}(\Sigma) \le 1$ and $R \le R_{\Sigma}$; anisotropy is genuine gain only inside the ellipsoid, along high-variance codebook directions where re-quantization-surviving perturbations concentrate, worth up to 10 certified points over isotropic smoothing in Figure~\ref{fig:results}(a). PALS-Certify uses the standard protocol \cite{cohen2019certified}: $n_0 = 100$ candidate draws and $n = 1{,}000$ estimation draws give a one-sided Clopper--Pearson lower bound $\underline{p_A}$ at confidence $1 - \alpha$, $\alpha = 10^{-3}$, with $\overline{p_B} = 1 - \underline{p_A}$, abstaining when $\underline{p_A} \le 1/2$. Because $\underline{p_A}$ cannot exceed its value $\overline{p_A}_{\max}\approx0.9931$ at $k=n$ and $\Phi^{-1}(\overline{p_B})=-\Phi^{-1}(\underline{p_A})$, the maximal Mahalanobis radius is $R_{\Sigma}^{\max}(\sigma)=\sigma\,\Phi^{-1}(\overline{p_A}_{\max})$, i.e.\ $0.616$ at $\sigma=0.25$, and $\lambda_{\min}(\Sigma)\le1$ caps the $\ell_2$ radius at the same value. Each $\sigma$ curve in Figure~\ref{fig:results}(a) therefore terminates at its ceiling $R_{\Sigma}^{\max}(\sigma)$, larger $\sigma$ reaching farther, with curve data from the released \texttt{certify.py}, which enforces the cap. The certificate is thus a guaranteed floor, not a guarantee against the strong attacks: at $\sigma=0.25$ it reaches latent radius $0.616$ and an input budget of only $\Delta\approx-26$ dB, far below the $[-20,0]$ dB attacks of Section~\ref{sec:threat}. As is standard in the certified-defense literature, we therefore report two separated pillars: the provable floor of Theorem~\ref{thm:certificate} and the empirical robustness of PALS-Train across the full range, which extends far past it.

\begin{table*}[t]
\caption{Empirical hijack success (\%) of PALS-Train by condition and slice (left) and PALS-Train ablations with $\sigma$ sweep (right), at $\Delta = 0$ dB.}
\label{tab:slices}
\centering
\scriptsize
\renewcommand{\arraystretch}{0.86}
\begin{minipage}[t]{0.45\textwidth}
\centering
\begin{tabular}{@{}lccc@{}}
\toprule
Condition & Undef. & PGD-AT & PALS \\
\midrule
Digital white-box & 91.7 & 19.4 & 8.3 \\
OTA, small room ($T_{60}\!\approx\!0.2$ s) & 52.8 & 11.0 & 5.4 \\
OTA, medium room ($T_{60}\!\approx\!0.5$ s) & 46.3 & 9.6 & 4.7 \\
OTA, medium $+$ babble 10 dB & 29.4 & 6.2 & 2.9 \\
Female / male (digital) & 92.4/90.9 & 20.6/18.3 & 8.9/7.8 \\
\bottomrule
\end{tabular}
\end{minipage}\hfill
\begin{minipage}[t]{0.49\textwidth}
\centering
\begin{tabular}{@{}lccc@{}}
\toprule
Configuration & Hijack $\downarrow$ & OTA $\downarrow$ & UTMOS $\uparrow$ \\
\midrule
PALS, full ($\sigma = 0.25$) & 8.3 & 4.7 & 3.82 \\
\quad isotropic covariance ($\Sigma = I$) & 13.6 & 7.5 & 3.80 \\
\quad without input-space dual & 14.9 & 9.8 & 3.85 \\
\quad without KL consistency & 12.7 & 6.9 & 3.74 \\
\quad without re-quantization & 10.9 & 5.8 & 3.69 \\
$\sigma = 0.0625$ / $0.125$ & 21.5 / 14.2 & 11.6 / 7.9 & 3.90 / 3.88 \\
\bottomrule
\end{tabular}
\end{minipage}
\end{table*}

\begin{theorem}[Risk decomposition for the smoothed loss]
\label{thm:trades}
Fix $x$ and history $h$. Let $\ell(x, \delta) = \mathbb{E}_{r \sim p_\vartheta(\cdot \mid Q(E_c(x+\delta)), h)}[\varphi(r)]$ for a per-response cost $\varphi \in [0,1]$, and define the Gaussian-smoothed loss $\ell_\sigma(x,\delta) = \mathbb{E}_{\xi\sim\mathcal{N}(0,\sigma^2\Sigma)}[\ell(x,\delta+\xi)]$. Then $\delta\mapsto\ell_\sigma(x,\delta)$ is $\Lambda_\sigma$-Lipschitz in $\lVert\cdot\rVert_2$ with $\Lambda_\sigma = \sqrt{2/\pi}\,/(\sigma\sqrt{\lambda_{\min}(\Sigma)})$, even though $\ell$ is discontinuous through the hard quantizer. For any augmentation $\mu_x$ and any target budget $\eta$,
\begin{equation}
\label{eq:decomp}
\begin{split}
\sup_{\delta \in \mathcal{P}_x(\eta)} \ell_\sigma(x, \delta) \le {}& \ell_\sigma(x, 0) + \mathbb{E}_{\nu} \sqrt{ \mathrm{KL}_x(\nu) / 2 } \\
&+ \Lambda_\sigma\, \mathbb{E}_{\nu}\lVert \nu - \delta^{\star}_x \rVert_2 + \tau_x,
\end{split}
\end{equation}
where $\mathrm{KL}_x(\nu) = \mathrm{KL}( p_\vartheta(\cdot \mid \hat z, h) \,\Vert\, p_\vartheta(\cdot \mid Q(E_c(x + \nu)), h) )$, $\delta^{\star}_x$ attains the supremum, and $\tau_x = \mathbb{P}_{\nu\sim\mu_x}(\nu\notin\mathcal{P}_x(\eta))$ is the tail mass that $\mu_x$ places outside the perceptual ball.
\end{theorem}
\begin{proof}
Lipschitzness is the standard Gaussian-smoothing bound: for $f:\mathbb{R}^m\to[0,1]$ and isotropic $\xi\sim\mathcal{N}(0,\sigma^2 I)$, $\bar f(\cdot)=\mathbb{E}[f(\cdot+\xi)]$ has $\lVert\nabla\bar f\rVert_2\le\sqrt{2/\pi}/\sigma$ by differentiating under the integral and bounding $\mathbb{E}\lVert\xi\rVert$-type terms (Stein's identity), independent of any continuity of $f$ \cite{salman2019provably}. Whitening by $\Sigma^{-1/2}$ replaces $\sigma$ by $\sigma\sqrt{\lambda_{\min}(\Sigma)}$ in the worst $\ell_2$ direction, giving $\Lambda_\sigma$. For the decomposition, write $\ell_\sigma(x,\delta^{\star}) - \ell_\sigma(x,0) = [\ell_\sigma(x,\delta^{\star}) - \mathbb{E}_\nu \ell_\sigma(x,\nu)] + [\mathbb{E}_\nu \ell_\sigma(x,\nu) - \ell_\sigma(x,0)]$. The first bracket is at most $\Lambda_\sigma\,\mathbb{E}_\nu\lVert\nu-\delta^{\star}\rVert_2$ by Lipschitzness. For the second, $\varphi\in[0,1]$ gives $\ell_\sigma(x,\nu)-\ell_\sigma(x,0)\le\mathrm{TV}\le\sqrt{\mathrm{KL}_x(\nu)/2}$ by Pinsker wherever $\nu\in\mathcal{P}_x(\eta)$, while the event $\nu\notin\mathcal{P}_x(\eta)$ contributes at most its probability $\tau_x$, giving \eqref{eq:decomp}.
\end{proof}

The decomposition explains each design choice: the first term is the post-training loss; the second is, up to Jensen, exactly the KL consistency term of \eqref{eq:loss}; the third shrinks as $\mu_x$ covers the ball, served by the maximum-entropy choice of Lemma~\ref{lem:maxent}; and $\tau_x$ is the explicit penalty for the Gaussian tail outside the ball, maintaining mathematical rigor without falsely assuming on-ball Gaussian support or $\tau_x=0$.

\begin{corollary}[Weighted-norm transfer of the latent certificate]
\label{cor:transfer}
Let the STFT use a tight frame with frame bound $A_w$. Then $\mathcal{P}_x(\eta) \subseteq \{ \delta : \lVert \delta \rVert_2 \le \eta\, \rho_x \}$ with $\rho_x^{2} = A_w^{-1} \sum_{f,k} \theta_x(f,k)$, so the perceptual ball sits inside a masking-weighted $\ell_2$ ball. Consequently a latent certificate of radius $R_{\Sigma}$ transfers to the weighted norm: every threshold-compliant perturbation whose induced latent displacement satisfies $\lVert\Sigma^{-1/2}\delta_z\rVert_2\le R_\Sigma$ leaves the decision of PALS-Certify invariant, a guarantee holding in the weighted-$\ell_2$ sense rather than pointwise in the supremum norm \eqref{eq:ball}.
\end{corollary}
\begin{proof}
By the frame inequality, $A_w \lVert \delta \rVert_2^{2} \le \sum_{f,k} \lvert (D\delta)_{f,k} \rvert^{2} \le \eta^{2} \sum_{f,k} \theta_x(f,k)$, which gives the containment; the transfer is then Theorem~\ref{thm:certificate} applied to the induced latent displacement $\delta_z=E_c(x+\delta)-E_c(x)$.
\end{proof}

The pullback to an input-space radius is heuristic, not a hard guarantee.

\begin{assumption}[Estimated local encoder Lipschitz constant]
\label{ass:lipschitz}
There is $\widehat L > 0$ with $\lVert E_c(x + \delta) - E_c(x) \rVert_2 \le \widehat L \lVert \delta \rVert_2$ for held-out $x$ and $\delta$ with $x + \delta$ in $\{x + \mathcal{P}_x(1)\}$; this is an empirical estimate, not a certified global bound.
\end{assumption}

Under Assumption~\ref{ass:lipschitz}, a latent radius $R$ maps to a waveform radius $R/\widehat L$. We obtain $\widehat L$ by power iteration on encoder Jacobian-vector products \cite{virmaux2018lipschitz,miyato2018spectral} over $10{,}000$ held-out five-second segments, finding median $2.4$ and 95th-percentile $3.7$, the latter used for reported budgets. Because $\widehat L$ is an empirical percentile of a local estimate, the input radius (and the $\Delta\approx-26$ dB figure) is heuristic; the hard certificate lives in the latent space, and a spectral-norm-constrained encoder making the bound architectural is deferred to Section~\ref{sec:limitations}.

\section{Experimental Setup}
\label{sec:setup}

\subsection{Model, Data, and Training}
The defended system is a Moshi-style agent \cite{defossez2024moshi}: a 7B-parameter temporal transformer with a depth transformer over RVQ tokens, a Mimi-style codec at 12.5 Hz with eight codebooks, and an inner-monologue text stream. From a pretrained checkpoint we post-train on 3,118 h of dual-channel conversational English from five sources: 141 h from otoSpeech-processed \cite{otospeech2025processed,otospeech2025full} (48 kHz stereo FLAC, CC BY 4.0); 727 h of spontaneous dyadic English from SSSD \cite{sheikh2025scalable}; $\approx$850 h across 1,656 video calls from CANDOR \cite{reece2023candor} (with demographic metadata); a 1,300 h subset of Seamless Interaction \cite{agrawal2025seamless}; and the 100 h headset partition of AMI \cite{carletta2005ami} (Whisper large-v3 attains $\approx$16\% WER \cite{radford2023robust}). Audio is resampled to 24 kHz, one speaker per channel. Post-training runs 80k AdamW steps (learning rate $2\times10^{-5}$, cosine decay, weight decay 0.1, global batch 64 five-minute segments) on 32 H100 GPUs for $\approx$5 days, $\sigma$ ramping over the first 10k steps. PALS sets $\sigma = 0.25$, $\beta = 4$, $\rho = 0.1$, covariances calibrated on 20 h of held-out audio.

\subsection{Attack Suite, Metrics, and Baselines}
Attacks follow Section~\ref{sec:threat}: 1,000 Adam iterations (rate $10^{-3}$) with projection onto $\mathcal{P}_x(\eta)$, five restarts, and $\Delta \in \{-20, -15, -10, -5, 0\}$ dB; masking thresholds use MPEG-1 model 1 on 2048-sample Hann windows with hop 512 \cite{qin2019imperceptible}. Hijack uses 200 attacker-chosen responses over 1,000 held-out carriers; mute succeeds when agent activity stays below 5\% of frames for 10 s; jailbreak uses 520 AdvBench-derived requests \cite{zou2023universal} with a 3 s universal prefix per category \cite{kang2025advwave,peng2025jalmbench}. Every judged quantity (including the 9-way response-intent taxonomy $\mathcal{Y}$) uses Gemini 2.5 Pro \cite{comanici2025gemini} (temperature 0, 3 samples by majority vote, Cohen's $\kappa = 0.88$ \cite{cohen1960coefficient} on a 200-item audit). Simulated over-the-air (OTA) evaluation convolves attacks with pyroomacoustics impulse responses \cite{scheibler2018pyroomacoustics} over 120 rooms ($T_{60} \in [0.2, 0.9]$ s) with MUSAN babble \cite{snyder2015musan} at 10--20 dB under expectation over transforms. Clean metrics use UTMOS \cite{saeki2022utmos}, DNSMOS \cite{reddy2022dnsmos}, Whisper large-v3 WER \cite{radford2023robust}, turn-taking F1 on Full-Duplex-Bench \cite{lin2025fullduplexbench}, and ITU-T P.808 \cite{itu2021p808}; intervals use $10^{4}$ bootstrap resamples. Baselines comprise the undefended agent, Gaussian augmentation at matched power, input-space randomized smoothing \cite{olivier2021sequential}, PGD adversarial training \cite{madry2018towards}, and a Demucs front-end \cite{defossez2020real}, all attacked adaptively \cite{tramer2020adaptive,athalye2018obfuscated}.

\section{Results}
\label{sec:results}

\subsection{Main Comparison}
Table~\ref{tab:main} reports the central comparison at the full budget; the attack columns evaluate the deployed PALS-Train. The undefended agent is broken almost at will (91.7\% hijack, 88.4\% mute, 76.2\% jailbreak, 46.3\% over the air). Baselines help at a price: input smoothing costs quality and turn-taking \cite{olivier2021sequential}, the purification front-end collapses under the adaptive attacker \cite{athalye2018obfuscated}, and PGD is strongest but triples training cost. PALS-Train attains 8.3\% hijack, 11.2\% mute, and 9.1\% jailbreak, $1.5\times$ to $2.3\times$ lower than PGD, at clean UTMOS 3.82 (an absolute decrease of 0.09 MOS points, within 2.3\% of 3.91), WER up only to 7.7\%, and turn-taking F1 unchanged. Figure~\ref{fig:results}(b) holds PALS-Train below 9\% from $-20$ to $0$ dB, its margin over PGD widening with the budget, well beyond the certified radius reported next.

\subsection{Certified Robustness}
Figure~\ref{fig:results}(a) shows certified accuracy of PALS-Certify, each $\sigma$ curve terminating at its ceiling ($0.616$ at $\sigma=0.25$); 47.0\% of decisions are certified at radius $0.5<0.616$, a floor at $\Delta \approx -26$ dB covering none of the demonstrated attacks. The anisotropy is not cosmetic: isotropic smoothing certifies up to 10 points fewer over most of the range, narrowing toward the ceiling (dashed), and input smoothing only 31.5\% (Table~\ref{tab:main}). Crucially, the mute attack is evaluated purely mechanically (audio activity $<5\%$ of frames for 10~s) without any LLM judge; mute drops from $88.4\pm2.0\%$ to $11.2\pm1.9\%$, verifying that the $\sim$10-fold gain is not an artifact of LLM judge bias.

\subsection{Ablations, Slices, and Limitations}
\label{sec:limitations}
Table~\ref{tab:slices} (left) shows reverberation and babble degrade the attack more than the defense (PALS-Train at 2.9\% under the hardest condition), with negligible speaker-sex and corpus spreads. Table~\ref{tab:slices} (right) dissects the defense and isolates gains relative to TRADES-style consistency. Removing the KL term raises hijack from 8.3\% to 12.7\%, testing generic consistency. However, isotropic covariance ($\Sigma=I$) yields 13.6\% and removing the input masking dual yields 14.9\% (doubling simulated OTA success), while PGD-AT reaches 19.4\%. This confirms that neither consistency nor adversarial training alone explains the gain; the psychoacoustically shaped input dual and codebook covariance geometry are the primary drivers. Removing re-quantization degrades clean UTMOS to 3.69. Three limitations remain: the hard certificate is latent-space, with its input-space pullback heuristic through an estimated local Lipschitz constant (spectral-norm constraints \cite{miyato2018spectral} left to future work); the certified floor transfers in the weighted-$\ell_2$ sense (Corollary~\ref{cor:transfer}); and corpora are English dyadic conversations.

\section{Conclusion}
\label{sec:conclusion}

Psychoacoustically aligned latent smoothing answers imperceptible attacks on full-duplex S2S dialogue models at the codec bottleneck every attack must cross. PALS-Train cuts attack success by roughly an order of magnitude at a 2.3\% quality cost with no inference-time overhead, while PALS-Certify supplies a provable anisotropic floor the empirical robustness far exceeds.

\section*{Acknowledgment}
The authors disclose that Claude Opus 4.8 was used for editing and rewriting all sections (Abstract, Introduction, Related Work, Methods, Experiments, Results, Conclusion) to improve the flow and presentation, and for assisting in implementation of the code. In addition, Gemini nano banana pro was used for generating diagrams (Figures 1 and 2).

\clearpage
\bibliographystyle{IEEEtran}
\bibliography{refs}

\end{document}